\documentclass{article}

 \usepackage[preprint]{neurips_2026}

\usepackage[utf8]{inputenc} 
\usepackage[T1]{fontenc}    
\usepackage{hyperref}       
\usepackage{url}            
\usepackage{booktabs}       
\usepackage{amsfonts}       
\usepackage{nicefrac}       
\usepackage{microtype}      
\usepackage[table]{xcolor}  
\usepackage{multirow}
\usepackage{colortbl}
\usepackage{amsmath}
\usepackage{amssymb}
\usepackage{amsthm}
\usepackage{enumitem}
\usepackage{graphicx}
\usepackage{wrapfig}
\graphicspath{{./images/}}
\graphicspath{{../plots/}{./figs/}}
\definecolor{lightcrimson}{rgb}{0.93, 0.16, 0.51}

\newtheorem{theorem}{Theorem}

\newtheorem{lemma}{Lemma}
\newtheorem{corollary}{Corollary}

\DeclareMathOperator{\Exp}{Exp}

\title{Mult-DPO: Multinomial Direct Preference Optimization for Recommender Systems}

\author{%
  Yaochen Zhu$^{1}$ \quad
  Harald Steck$^{2}$ \quad
  James McInerney$^{2}$ \quad
  Aditya Sinha$^{2}$ \\ \bf
  Yinhan He$^{1}$ \quad
  Nathan Kallus$^{2,3}$ \quad
  Jundong Li$^{1}$ \\[2pt] \bf
  $^{1}$University of Virginia \quad
  $^{2}$Netflix \quad
  $^{3}$Cornell University \\[2pt] \bf
  \texttt{\{uqp4qh, nee7ne, jundong\}@virginia.edu} \\ \bf
  \texttt{\{hsteck, jmcinerney, adityasinha, nkallus\}@netflix.com}
}

\begin{document}

\maketitle

\begin{abstract}

Direct preference optimization (DPO) is a simple and effective alignment strategy for large language models (LLMs) based on pairwise preferences. In recommender systems, however, user feedback is rarely pairwise. For a given context, e.g., a user, a session, or a conversation, we typically observe \textbf{set-wise preferences} with \emph{multiple} positive items, where every positive item should outrank every unobserved or explicitly negative item, with no prescribed order among the positives or the negatives themselves. A natural generalization is to use the Plackett–Luce (PL) reward model, which extends the Bradley–Terry reward model underlying vanilla DPO from pairwise preferences to full rankings of candidates. However, we show that adapting the PL model to set-wise preferences requires marginalizing over all positive orderings, where the resulting expression is combinatorial in complexity. To address this fundamental challenge, we propose \textbf{Mult-DPO}, a novel DPO objective with a tractable multinomial surrogate likelihood over set-wise preference events for the user-preference alignment of LLM-based recommender systems. The multinomial construction is not itself a ranking distribution, but it is defined on the same reward-induced weight space and admits a closed-form DPO-style objective, enabling direct alignment of LLMs with multiple candidates through a classification-style objective. In addition, we prove that the multinomial DPO loss is a tractable \emph{upper bound} on the marginalized PL DPO loss when optimizing against the set-wise preference data. We further characterize the tightness of this bound in terms of the relative total weight of positives versus negatives, which provides insights into tightening the bound with richer or harder negatives. Finally, we extend Mult-DPO to the alignment of LLMs with multiple preference levels. Code is available at {\color{lightcrimson}\url{https://github.com/yaochenzhu/Mult_DPO}}.

\end{abstract}

\section{Introduction}

Large language models (LLMs) have achieved remarkable success through large-scale pretraining. To further align these models with human preferences, reinforcement learning from human feedback (RLHF) \citep{christiano2017deep, ouyang2022training} emerged as a popular framework: It first trains a reward model from pairwise human preferences, and then optimizes the LLM policy via reinforcement learning (RL). Recently, direct preference optimization (DPO) \citep{rafailov2023direct} offers a simpler alternative by leveraging a closed-form connection between the optimal policy and the reward function in RLHF, eliminating the RL stage and reducing the alignment to a classification-style objective on preference data. Owing to its conceptual simplicity and strong empirical performance, DPO has been widely adopted in LLM alignment tasks such as summarization \citep{liu2023statistical}, code generation \citep{weyssow2024codeultrafeedback}, and mathematical reasoning \citep{pang2024iterative}.

As LLMs increasingly serve as the backbone of recommender systems (RSs) \citep{bao2023tallrec, liao2024llara, he2023zeroshot}, it is natural to ask whether DPO can play the same role in aligning LLM-based RSs with user preferences. The answer, however, is complicated by a fundamental mismatch between the preference data structure assumed by vanilla DPO and that observed in recommendation tasks. Specifically, DPO assumes \emph{pairwise} preferences between two candidates per context, which suits the setting of long responses for QA-style tasks where candidates are difficult to generate \citep{bai2022training, stiennon2020summarization}. In contrast, user feedback in RSs is rarely pairwise. For a given context in RSs, e.g., an interaction history or a conversation between user and RS, we typically observe \emph{multiple} positive items (e.g., items liked, clicked by the user) and negative items (e.g., items unobserved or explicitly rejected) \citep{rendle2009bpr, liang2018vae}. From the perspective of preference alignment, every positive should be preferred by the LLM over every negative, yet no order should be enforced among the positives or among the negatives themselves. Applying vanilla DPO to set-wise preferences therefore requires enumerating all positive-negative combinations, which is computationally expensive and discards the joint rank structure of the preference data.

Recently, several works have adapted DPO to preference alignment of LLMs with multiple candidates, which can be categorized into two technical routes. One route is to generalize the Bradley-Terry (BT) \citep{bradley1952rank} reward model underlying vanilla DPO from pairwise preferences to its listwise counterpart, the Plackett-Luce (PL) model \citep{plackett1975analysis, luce1959individual}, including methods such as PRO \citep{song2024preference}, KPO \citep{zhang2025kpo}, etc. However, when applied to set-wise preference, the PL likelihood must be marginalized over all positive orderings consistent with the observation (where the orderings of the negatives can be directly summed-out), which is combinatorially expensive. Another route sidesteps this intractability by restricting the supervision to one single positive per context paired with multiple negatives. For example, DMPO \citep{bai2024dmpo} contrasts the positive against the arithmetic mean of the negative log-ratios within the BT model, whereas S-DPO \citep{chen2024softmax} shows that, under the single-positive restriction, the marginalized PL likelihood collapses to a closed-form softmax DPO loss that contrasts the positive jointly against all sampled negatives. Nevertheless, neither route faithfully retains the joint set-wise structure of the feedback, leaving alignment with multi-positive preferences a fundamentally unresolved challenge.

To address the challenge, we propose \textbf{Mult-DPO}, a novel DPO-style objective for aligning LLM-based RSs with set-wise preferences. We first introduce a multinomial surrogate event model as a tractable lower-bound surrogate for the marginalized PL likelihood over set-wise preferences between the positive and negative candidate sets. Although this surrogate is not a ranking model, it is defined on the same reward-induced weights as BT and PL and yields a closed-form DPO-style objective through the standard RLHF policy-ratio reparameterization. Theoretically, we prove that minimizing the multinomial DPO loss optimizes a tractable upper bound on the otherwise intractable marginalized PL DPO loss. Furthermore, we characterize the tightness of this bound in closed form in terms of the relative total weights (i.e., exponentiated rewards) of positives versus negatives, and show that the bound becomes tighter when richer or harder negatives are selected. Finally, we extend Mult-DPO to alignment with multiple preference levels when fine-grained user preference data is available (e.g., explicit ratings)  via a \textit{sequential multinomial (SMN)} surrogate, and show that all theoretical properties of Mult-DPO from the two-set case carry over to this generalization. On both general recommendation and conversational recommendation benchmarks, we demonstrate that Mult-DPO and its multi-level extension consistently outperform various DPO-based baselines.

\vspace{-2mm}

\section{Preliminaries}

\vspace{-1mm}

\subsection{Problem Formulation}
\label{sec:problem}

In this paper, we aim to tackle the fundamental challenge of aligning LLM-based RSs with user preferences over multiple candidates. Specifically, let $x$ denote the recommendation context, which may include a user profile, an interaction history, or a conversation, and let $\mathcal{C}$ denote the item catalog from which recommendations should be generated, where each item $e$ is rendered as a token sequence $y(e)=(y_1,\ldots,y_{m_e})$. For each context $x$ in the preference dataset $\mathcal{D}$, we observe user preference over candidate items in the form of $(x, \mathcal{E}^{p}, \mathcal{E}^{d})$, where $\mathcal{E}^{p} = \{e_1, \ldots, e_k\}$ and $\mathcal{E}^{d} = \{e_{k+1}, \ldots, e_K\}$ denote disjoint positive and negative item sets associated with $x$, respectively. Here, $\mathcal{E}^{p}$ contains items receiving positive feedback (e.g., liked or clicked), whereas $\mathcal{E}^{d}$ contains uninteracted or rejected items. We use $\mathcal{E} = \mathcal{E}^{p} \cup \mathcal{E}^{d} \subset \mathcal{C}$ to denote the full candidate set, with $k=|\mathcal{E}^{p}|$ and $K-k =|\mathcal{E}^{d}|$ for the set-wise contrastive setting studied below. We note that the user preference tuple implies a \textbf{set-wise ranking constraint} that every positive item in $\mathcal{E}^{p}$ should outrank every negative item in $\mathcal{E}^{d}$:
\begin{equation}
\label{eq:set_rank}
\Omega_x := \{ e \succ e' \mid e \in \mathcal{E}^{p}, \ e' \in \mathcal{E}^{d} \},
\end{equation}
but no order is implied among the positives or the negatives themselves. Let $\pi_{\theta}(e \mid x)$ denote an LLM-based RS policy that assigns a generation probability to each candidate item $e$ given the context $x$. This paper aims to align $\pi_{\theta}$ with the set-wise preference structure in $\Omega_x$ that fully leverages the joint multi-positive and negative structure while remaining computationally efficient, such that the items in catalog $\mathcal{C}$ can be effectively ranked for recommendation at inference time.



\subsection{Theoretical Analysis: RLHF, DPO and Marginalized Plackett-Luce Extension}
\label{sec:RLHF}

In this section, we first review reinforcement learning from human feedback and direct preference optimization, and discuss the naive extension of DPO to set-wise preferences with a marginalized Plackett-Luce likelihood, which forms the foundation of Mult-DPO introduced in Section~\ref{sec:method}.

\textbf{2.2.1 Reinforcement Learning from Human Feedback (RLHF).}
RLHF aligns an LLM-based policy $\pi_{\theta}(e \mid x)$ with user preferences by optimizing it against a reward model derived from human preferences. Specifically, for each context-item pair $(x, e)$, we assume a latent reward $r(x, e) \in \mathbb{R}$ indicating the utility of recommending item $e$ under context $x$. We define the associated weight as $w(e \mid x) := \exp(r(x,e)) > 0$ for the convenience of linking $r(x, e)$ to probabilistic distributions. The \textbf{Bradley-Terry Preference Model} (BT) ~\citep{bradley1952rank} assumes that, given two candidate items $e_p$ and $e_d$ for context $x$ with $e_p$ preferred over $e_d$, the probability of observing such pair-wise preference can be formulated with the difference between their latent rewards as follows:
\begin{equation}
P(e_p \succ e_d \mid x)
=
\sigma\!\big(r(x,e_p) - r(x,e_d)\big)
=
\frac{w(e_p \mid x)}{w(e_p \mid x) + w(e_d \mid x)},
\label{eq:bt_model}
\end{equation}
where $\sigma(\cdot)$ is the sigmoid function. The reward model is then estimated by minimizing the negative log-likelihood of the observed comparisons induced by Eq. (\ref{eq:bt_model}). Given a fixed reference policy $\pi_{\mathrm{ref}}(e \mid x)$ and learned reward model $r(x,e)$, the RLHF objective can be formulated as follows:
\begin{equation}
\max_{\pi(\cdot \mid x)}
\ \mathbb{E}_{e \sim \pi(\cdot \mid x)}[r(x,e)]
-
\beta \, \mathrm{KL}\!\left(
\pi(\cdot \mid x)\,\|\,\pi_{\mathrm{ref}}(\cdot \mid x)
\right),
\label{eq:multdpo_rlhf}
\end{equation}
where $\beta > 0$ is an important hyperparameter that controls the regularization strength.

\vspace{1mm}
\noindent \textbf{2.2.2 Direct Preference Optimization (DPO).}
DPO~\citep{rafailov2023direct} avoids explicit reward modeling by observing that the RLHF objective in Eq.~\eqref{eq:multdpo_rlhf} admits a closed-form solution as follows:
\begin{equation}
r(x,e)
=
\beta \log \frac{\pi^{*}(e \mid x)}{\pi_{\mathrm{ref}}(e \mid x)}
+
\beta \log Z(x),
\label{eq:multdpo_reward_policy_ratio}
\end{equation}
where $Z(x)$ is an intractable partition function. Observing that the BT likelihood in Eq.~\eqref{eq:bt_model} depends only on reward differences, $Z(x)$ cancels out after substituting Eq.~\eqref{eq:multdpo_reward_policy_ratio}. Replacing $\pi^*$ with a trainable policy $\pi_\theta$ yields the DPO objective as follows:
\begin{equation}
\mathcal{L}_{\mathrm{DPO}}(x,e_p,e_d)
=
-\log
\sigma\!\left(
\beta \log \frac{\pi_\theta(e_p \mid x)}{\pi_{\mathrm{ref}}(e_p \mid x)}
-
\beta \log \frac{\pi_\theta(e_d \mid x)}{\pi_{\mathrm{ref}}(e_d \mid x)}
\right),
\label{eq:multdpo_pairwise_dpo}
\end{equation}
which reduces RLHF alignment to a classification-like objective on the preference data. 

\vspace{1mm}
\textbf{2.2.3 Marginalized Plackett-Luce DPO Objective for Set-wise Preference.} Vanilla DPO can only align LLMs with two candidates due to the pairwise constraint of the BT preference model. A natural generalization of BT to multiple candidates is the Plackett-Luce (PL) model~\citep{plackett1975analysis,luce1959individual}. As in BT, PL associates each item $e \in \mathcal{E}$ with a latent reward $r(x, e)$ (and weight $w(e \mid x)$). For a permutation $\tau$ of the candidate items in $\mathcal{E}$, the PL likelihood can be formulated as follows:
\begin{equation}
p_{\mathrm{PL}}(\tau \mid x,\mathcal{E};w)
=
\prod_{t=1}^{|\mathcal{E}|}
\frac{w(e_{\tau(t)} \mid x)}
{\sum_{j=t}^{|\mathcal{E}|} w(e_{\tau(j)} \mid x)},
\label{eq:pl_full_rank}
\end{equation}
which denotes a sequential selection process where the item at rank $t$ is sampled with probability proportional to its weight among the remaining candidates. To use PL to model the set-wise preference event $\Omega_x$, however, since the positive and negative ordering are unknown, we must marginalize over all consistent ranks of $\mathcal{E}$. Let $S_k$ denote the permutations of the positive indices $\{1, \ldots, k\}$, and define the cumulative weights of the positive set, negative set, and full candidate set under context $x$:
\begin{equation}
A := \sum_{e \in \mathcal{E}^{p}} w(e \mid x), \ \
B := \sum_{e \in \mathcal{E}^{d}} w(e \mid x), \ \
W := A + B = \sum_{e \in \mathcal{E}} w(e \mid x),
\label{eq:multdpo_mass_def}
\end{equation}
With these definitions, we show that by marginalizing the PL likelihood in Eq. (\ref{eq:pl_full_rank}) (see Appendix~\ref{appendix:mult-dpo}), the likelihood of $\Omega_x$ can be formulated as follows (which we dub the \textbf{marginalized PL event model}):
\begin{equation}
p_{\mathrm{PL}}(\Omega_x \mid x, \mathcal{E}; w)
=
\sum_{\rho \in S_k}
\prod_{t=1}^{k}
\frac{w(e_{\rho(t)} \mid x)}
{B + \sum_{j=t}^{k} w(e_{\rho(j)} \mid x)},
\label{eq:multdpo_pl_likelihood}
\end{equation}
where the marginalization over the negative permutations gracefully vanishes into the cumulative weights $B$. An equivalent inclusion-exclusion form derived from the exponential-race interpretation of PL reduces the number of terms from $k!$ to $2^k$ (see \ref{app:pl-ie}), but it still remains exponential w.r.t. $k$. Therefore, directly optimizing Eq.~\eqref{eq:multdpo_pl_likelihood} as a DPO-style objective is infeasible even if $k$ is moderate.

\section{Mult-DPO: Multinomial Direct Preference Optimization}
\label{sec:method}

\subsection{Multinomial Surrogate Event Model}
\label{sec:mult}

To obtain an effective and tractable surrogate for the exact marginalized PL event model for aligning LLMs with set-wise preferences, we propose an alternative event model defined on the same weight space, i.e., $w(e \mid x), \ \forall e \in \mathcal{E}$, which we name the \textbf{multinomial (MN) surrogate}. The surrogate first normalizes the weights into a categorical distribution over candidates:
\begin{equation}
p(e \mid x) := \frac{w(e \mid x)}{W}, \ \ e \in \mathcal{E}.
\label{eq:multdpo_categorical}
\end{equation}
The surrogate likelihood then defines the observed set-wise event $\Omega_x$ as the probability, under $k$ independent item draws from $p(\cdot \mid x)$, that \textit{each positive item in $\mathcal{E}^{p}$ appears exactly once while no negative item is sampled}. Since the $k$ positives can appear in any draw order, there are $k!$ equivalent sequences that correspond to this event. Therefore, the MN surrogate likelihood of the set-wise preference event can be formulated as follows:
\begin{equation}
p_{\mathrm{MN}}(\Omega_x \mid x, \mathcal{E}; w)
=
k! \prod_{e \in \mathcal{E}^{p}} \frac{w(e \mid x)}{W},
\label{eq:multdpo_mn_likelihood}
\end{equation}
which can be computed in $\mathcal{O}(k)$ complexity.
Unlike PL, the MN construction is not a distribution over permutations: it is an IID event likelihood and assigns probability mass to duplicate draw sequences outside the valid ranking space. We use it as a tractable surrogate for the observed positive set, and the lower-bound result justifies it as a conservative surrogate for the exact marginalized PL likelihood:
\begin{theorem}
\label{thm:multdpo_lb}
For any disjoint sets $\mathcal{E}^{p}$ and $\mathcal{E}^{d}$ with $|\mathcal{E}^{p}|\ge 1$, and any positive weights $\{w(e \mid x)\}_{e \in \mathcal{E}}$,
\begin{equation}
p_{\mathrm{PL}}(\Omega_x \mid x, \mathcal{E}; w)
\ge
p_{\mathrm{MN}}(\Omega_x \mid x, \mathcal{E}; w).
\label{eq:multdpo_lb}
\end{equation}
\end{theorem}
In addition, we prove that the ratio between the PL likelihood and the MN surrogate likelihood is bounded, which will later imply a loss-gap bound between marginalized PL-DPO and Mult-DPO.
\begin{theorem}
\label{thm:multdpo_tightness}
For any disjoint sets $\mathcal{E}^{p}$ and $\mathcal{E}^{d}$ with $|\mathcal{E}^{p}|\ge 1$ and $|\mathcal{E}^{d}|\ge 1$, and any positive weights $\{w(e \mid x)\}_{e \in \mathcal{E}}$, the ratio between the marginalized PL likelihood and the multinomial surrogate likelihood is bounded by the following relations
\begin{equation}
1
\;\le\;
\frac{p_{\mathrm{PL}}(\Omega_x \mid x, \mathcal{E}; w)}{p_{\mathrm{MN}}(\Omega_x \mid x, \mathcal{E}; w)}
\;\le\;
\left(1+\frac{A}{B}\right)^{k-1}.
\label{eq:multdpo_tightness}
\end{equation}
\end{theorem}
The proof of Theorems~\ref{thm:multdpo_lb} and~\ref{thm:multdpo_tightness} is in \ref{sec:proof1} and~\ref{sec:proof2}. We note that the bound in Eq.~\eqref{eq:multdpo_tightness} is exact when $k=1$, in which case the MN surrogate likelihood reduces to the categorical likelihood in S-DPO \citep{chen2024softmax}, which generalizes DPO to alignment with \textbf{one} positive over multiple negatives.

\subsection{Mult-DPO Objective for LLM Alignment with Set-wise Preference}
\label{sec:mult-objective}

Building on the MN surrogate, we further show that it is compatible with the same RLHF reward-policy reparameterization that underlies DPO, yielding a classification-like alignment objective on (now set-wise) preference data. Recall that, under RLHF, the reward can be expressed in terms of the optimal policy and the reference policy as $r(x,e)=\beta \log \frac{\pi^{*}(e \mid x)}{\pi_{\mathrm{ref}}(e \mid x)}+\beta \log Z(x)$ (see Eq.~\eqref{eq:multdpo_reward_policy_ratio}), and the item weight is defined as $w(e \mid x):=\exp(r(x,e))$. Combining these two relations and replacing the optimal policy $\pi^*$ with a trainable policy $\pi_\theta$ yields the \textbf{policy-induced weights} as follows:
\begin{equation}
w_{\boldsymbol{\pi_\theta}}(e \mid x)
\propto
\left(
\frac{\pi_\theta(e \mid x)}{\pi_{\mathrm{ref}}(e \mid x)}
\right)^{\beta},
\label{eq:multdpo_dpo_weight}
\end{equation}
where the omitted proportionality constant $Z_{\pi_{\theta}}(x)^{\beta}$ is shared across candidates and cancels in both the PL and MN surrogate likelihoods. Substituting the policy-induced weights $w_{\pi_\theta}(e \mid x)$ into Eq.~\eqref{eq:multdpo_mn_likelihood} and taking the negative log-likelihood leads to the proposed \textbf{Mult-DPO} objective as follows:
\begin{align}
\mathcal{L}_{\mathrm{Mult\text{-}DPO}}(x,\mathcal{E}^{p},\mathcal{E}^{d})
&:=
-\log p_{\mathrm{MN}}(\Omega_x \mid x, \mathcal{E}; w_{\pi_\theta}) \nonumber \\
&\;=\;
-\beta \sum_{e \in \mathcal{E}^{p}}
\log \frac{\pi_\theta(e \mid x)}{\pi_{\mathrm{ref}}(e \mid x)}
+
k \log \sum_{e \in \mathcal{E}}
\left(
\frac{\pi_\theta(e \mid x)}{\pi_{\mathrm{ref}}(e \mid x)}
\right)^{\beta}
+ C.
\label{eq:multdpo_final_obj}
\end{align}
where $C$ is a constant depending only on $k$ and can be omitted during optimization. 


\subsection{Relation between Marginalized PL and Multinomial DPO Objective}
To further characterize the relation between the PL-DPO and the tractable Mult-DPO objectives, we define the policy-induced cumulative weights of the positives and negatives analogous to Eq.~\eqref{eq:multdpo_mass_def} as:
\begin{equation}
A_{\pi_\theta}
:=
\sum_{e \in \mathcal{E}^{p}} w_{\pi_\theta}(e \mid x), \ \
B_{\pi_\theta}
:=
\sum_{e \in \mathcal{E}^{d}} w_{\pi_\theta}(e \mid x),
\label{eq:multdpo_policy_masses}
\end{equation}
where both $A$ and $B$ now depend on policy $\pi_\theta$. Substituting the policy-induced weights from Eq.~\eqref{eq:multdpo_dpo_weight} into Eq.~\eqref{eq:multdpo_pl_likelihood} and taking the negative log-likelihood gives the ideal but intractable PL-DPO loss:
\begin{equation}
\mathcal{L}_{\mathrm{PL\text{-}DPO}}(x,\mathcal{E}^{p},\mathcal{E}^{d})
:=
-\log p_{\mathrm{PL}}(\Omega_x \mid x, \mathcal{E}; w_{\pi_\theta}),
\label{eq:multdpo_pl_dpo_loss}
\end{equation} 
Theorems~\ref{thm:multdpo_lb} and~\ref{thm:multdpo_tightness} immediately imply the following corollary between $\mathcal{L}_{\mathrm{PL\text{-}DPO}}$ and $\mathcal{L}_{\mathrm{MN\text{-}DPO}}$.

\begin{corollary}
\label{cor:multdpo_dpo_bound}
For every context $x$ and every training instance $(x,\mathcal{E}^{p},\mathcal{E}^{d})$,
\begin{equation}
\mathcal{L}_{\mathrm{PL\text{-}DPO}}(x,\mathcal{E}^{p},\mathcal{E}^{d})
\le
\mathcal{L}_{\mathrm{Mult\text{-}DPO}}(x,\mathcal{E}^{p},\mathcal{E}^{d}),
\label{eq:multdpo_dpo_order}
\end{equation}
and
\begin{equation}
0
\le
\mathcal{L}_{\mathrm{Mult\text{-}DPO}}(x,\mathcal{E}^{p},\mathcal{E}^{d})
-
\mathcal{L}_{\mathrm{PL\text{-}DPO}}(x,\mathcal{E}^{p},\mathcal{E}^{d})
\le
(k-1)\log\!\left(1+\frac{A_{\pi_\theta}}{B_{\pi_\theta}}\right).
\label{eq:multdpo_dpo_tightness}
\end{equation}
\end{corollary}
The proof follows from Theorem~\ref{thm:multdpo_tightness} by substituting the policy-induced weights and taking negative logarithms. The corollary shows that the Mult-DPO loss is a tractable upper bound on the exact but otherwise intractable marginalized PL-DPO loss. In addition, it reveals that, for a fixed $k$-positive set and current policy, increasing the non-negligible policy-induced weight $\pi_\theta$ for the negatives $B_{\pi_\theta}$ sharpens the worst-case gap, suggesting that the MN surrogate becomes a better approximation of the intractable marginalized PL-DPO loss when richer or harder negatives are selected as candidates.

\subsection{Extension of Mult-DPO to Multi-Level Preference}
\label{sec:smn}

The preceding sections assume binary preferences from users, which cover the most common RS setting, including recommendation with implicit feedback \citep{rendle2009bpr}. In practice, however, user feedback can be more fine-grained. For example, users may provide different ratings for the items. Such feedback naturally induces a \textbf{multi-level preference structure}, where items are partitioned into multiple preference groups. To capture this richer preference structure, we extend Mult-DPO from binary set-wise preferences to \underline{mult}i-level set-wise preferences, denoted \textbf{Mult$^2$-DPO}.

\vspace{1mm}
\noindent \textbf{3.4.1 Problem Formulation of Multi-Level Preference.}
Suppose that, for a recommendation-relevant context $x$, the candidate set $\mathcal{E}$ is partitioned into $G\ge 2$ disjoint preference groups as:
\begin{equation}
\mathcal{E}
=
\bigcup_{g=1}^{G} \mathcal{E}^{(g)}, \ \
|\mathcal{E}^{(g)}| = k_g, \ \
\sum_{g=1}^{G} k_g = K,
\label{eq:multdpo_group_partition}
\end{equation}
where the group index denotes the preference level, i.e., items in $\mathcal{E}^{(g)}$ should outrank items in $\mathcal{E}^{(h)}$ if $g < h$, with no order among items within each group. We denote the multi-level preference event as:
\begin{equation}
\Omega^{\mathrm{grp}}_x
:=
\Big\{
e \succ e'
\ \big|\
e \in \mathcal{E}^{(g)},\
e' \in \mathcal{E}^{(h)},\
1 \le g < h \le G
\Big\}.
\label{eq:multdpo_group_event}
\end{equation}
For each boundary $g=1,\ldots,G-1$, we further define the \textbf{group-$g$-specific preference event} $\Omega^{(g)}_x := \{ e \succ e' \mid e \in \mathcal{E}^{(g)},\ e' \in \bigcup_{h=g+1}^{G} \mathcal{E}^{(h)} \}$. This event has the same form as the two-set preference event $\Omega_x$ in Eq.~\eqref{eq:set_rank} if we view $\mathcal{E}^{(g)}$ as $\mathcal{E}^{p}$ and
$\bigcup_{h=g+1}^{G} \mathcal{E}^{(h)}$ as $\mathcal{E}^{d}$. By definition, we have $\Omega^{\mathrm{grp}}_x = \bigcap_{g=1}^{G-1} \Omega^{(g)}_x$.
For each boundary $g=1,\ldots,G-1$ and any positive weights, we define
\begin{equation}
A_g := \sum_{e \in \mathcal{E}^{(g)}} w(e \mid x), \ 
B_g := \sum_{h=g+1}^{G}\sum_{e \in \mathcal{E}^{(h)}} w(e \mid x), \ 
W_g := A_g + B_g
=
\sum_{h=g}^{G}\sum_{e \in \mathcal{E}^{(h)}} w(e \mid x),
\label{eq:multdpo_group_masses}
\end{equation}
where $A_g$ is the total weight of group $g$ and $B_g$ is the total weight of all lower-preference groups.

\vspace{1mm}
\noindent \textbf{3.4.2 Marginalized PL Reward Model for Multi-Level Preference.}
By the sequential selection property of PL, once items from the higher-preference groups $\{\mathcal{E}^{(1)}, \ldots, \mathcal{E}^{(g-1)}\}$ have been placed, the events $\Omega^{(g)}_x$ become conditionally independent across $g$, which yields the following recursive factorization of the marginalized PL likelihood for the multi-level preference data:
\begin{equation}
p_{\mathrm{PL}}(\Omega^{\mathrm{grp}}_x \mid x, \mathcal{E}; w)
=
\prod_{g=1}^{G-1}
p_{\mathrm{PL}}
\!\left(
\Omega^{(g)}_x
\ \middle|\
x,\ \bigcup_{h=g}^{G} \mathcal{E}^{(h)};\ w
\right),
\label{eq:multdpo_group_pl_factorization}
\end{equation}
where each factor is a set-wise marginalized PL likelihood of the same form studied in Section~\ref{sec:mult}.

\vspace{1mm}
\noindent \textbf{3.4.3 Sequential Multinomial (SMN) Surrogate Model and SMN-DPO.}
This factorization in Eq.~\eqref{eq:multdpo_group_pl_factorization} naturally extends the multinomial lower bound group by group. Applying the two-set MN surrogate to each group-specific event yields the sequential multinomial (SMN) surrogate
\begin{equation}
p_{\mathrm{SMN}}(\Omega^{\mathrm{grp}}_x \mid x, \mathcal{E}; w)
:=
\prod_{g=1}^{G-1}
\left(
k_g!
\prod_{e \in \mathcal{E}^{(g)}} \frac{w(e \mid x)}{W_g}
\right),
\label{eq:multdpo_group_smn}
\end{equation}
By applying Theorem~\ref{thm:multdpo_lb} at each factor, Eq.~\eqref{eq:multdpo_group_smn} provides a tractable lower bound on the exact PL likelihood in Eq.~\eqref{eq:multdpo_group_pl_factorization}. Substituting the policy-induced weight $w_{\pi_\theta}(e \mid x)$ from Eq.~\eqref{eq:multdpo_dpo_weight} into Eq.~\eqref{eq:multdpo_group_smn} and taking the negative log-likelihood leads to the multi-group Mult-DPO (\textbf{Mult$^{2}$-DPO}) objective
\begin{align}
\mathcal{L}_{\text{Mult}^{2}\text{-DPO}}
&:=
-\log p_{\mathrm{SMN}}(\Omega^{\mathrm{grp}}_x \mid x, \mathcal{E}; w_{\pi_\theta})
\nonumber\\
&\;=\;
\sum_{g=1}^{G-1}
\left[
-\beta\sum_{e \in \mathcal{E}^{(g)}}
\log\frac{\pi_\theta(e \mid x)}{\pi_{\mathrm{ref}}(e \mid x)}
+
k_g \log
\sum_{h=g}^{G}
\sum_{e \in \mathcal{E}^{(h)}}
\left(\frac{\pi_\theta(e \mid x)}{\pi_{\mathrm{ref}}(e \mid x)}\right)^{\beta}
\right]
+ C',
\label{eq:multdpo_group_final}
\end{align}
where $C' = -\sum_{g=1}^{G-1} \log k_g!$ is a constant independent of $\theta$ and can be omitted during optimization. When $G = 2$, Eq.~\eqref{eq:multdpo_group_final} reduces to the binary Mult-DPO objective in Eq.~\eqref{eq:multdpo_final_obj}. The same boundary-wise reduction also gives a multi-level analogue of Corollary~\ref{cor:multdpo_dpo_bound}. If we define the PL-DPO loss
\begin{equation}
\mathcal{L}_{\mathrm{PL\text{-}DPO}}^{\mathrm{grp}}
:=
-\log p_{\mathrm{PL}}(\Omega^{\mathrm{grp}}_x \mid x, \mathcal{E}; w_{\pi_\theta}),
\label{eq:multdpo_group_pl_dpo_loss}
\end{equation}
and let $A_{\pi_\theta,g}$ and $B_{\pi_\theta,g}$ denote the group masses in Eq.~\eqref{eq:multdpo_group_masses} after substituting $w_{\pi_\theta}$ for $w$, we have:
\begin{corollary}[Multi-level loss-gap bound]
\label{cor:multdpo_group_dpo_bound}
For every context $x$ and every multi-level preference instance with nonempty ordered groups, we have
\begin{equation}
0
\le
\mathcal{L}_{\mathrm{Mult}^{2}\mathrm{-DPO}}
-
\mathcal{L}_{\mathrm{PL\text{-}DPO}}^{\mathrm{grp}}
\le
\sum_{g=1}^{G-1}
(k_g-1)
\log\!\left(1+\frac{A_{\pi_\theta,g}}{B_{\pi_\theta,g}}\right).
\label{eq:multdpo_group_dpo_tightness}
\end{equation}
\end{corollary}

\subsection{Complexity Analysis}
\label{sec:complexity}

Although Mult-DPO aligns LLM-based RSs with user preferences over more candidates per training step than its single-positive counterparts (e.g., vanilla DPO and S-DPO), the per-step time complexity remains comparable thanks to the shared prompt prefix that allows \textbf{KV-cache reuse} of the prompt across all candidates. Concretely, let $N_{x}$ denote the number of prompt tokens and $N_{i}$ the average number of tokens per item. With KV-cache reuse, the per-step self-attention cost of every DPO-style method takes the form $\mathcal{O}(N_x^2 + cN_xN_i + cN_i^2)$. Here, $c$ is the number of candidates scored per step: $c = 2$ for vanilla DPO, $c = 1 + (K - k)$ for S-DPO, and $c = K$ for Mult-DPO, where $k$ and $K$ are the numbers of positives and full candidates, respectively. Therefore, when $K \ll N_{x}$, which is typical in the RS settings as the prompt usually contains user context, user/item features, interaction history, or conversation, the prompt-level self-attention term $N_{x}^{2}$ dominates; in our experiments, we indeed find that Mult-DPO exhibits similar wall-clock times to vanilla DPO and S-DPO.

\section{Empirical Study}
\label{sec:experiments}

In this section, we evaluate Mult-DPO and its extension to multi-level preference alignment, i.e., Mult$^{2}$-DPO, on both general and conversational recommendation tasks, aiming to answer the following three research questions: \textbf{(RQ1)} how does Mult-DPO perform compared to other supervised finetuning (SFT)-based baselines or DPO-style baselines for aligning LLM-based RSs under set-wise feedback; \textbf{(RQ2)} how well does the MN DPO loss of Mult-DPO match the theoretical upper bound on the marginalized PL DPO loss from Corollary~\ref{cor:multdpo_dpo_bound}; and \textbf{(RQ3)} does Mult$^{2}$-DPO yield further performance gains over Mult-DPO when explicit multi-level preferences are available?

\subsection{Datasets and Experimental Setup}
\label{sec:exp_setup}

For general recommendation tasks, we use the \texttt{MovieLens-10M}~\citep{harper2015movielens} and \texttt{Goodreads}~\citep{wan2018goodreads} datasets. For \texttt{MovieLens-10M}, where explicit user-item ratings are available, we first treat items rated $5$ as positives and randomly sampled unrated items as negatives, and then treat different ratings as different preference levels in Section~\ref{sec:exp_multlevel} to evaluate Mult$^2$-DPO. For conversational recommendation, we adapt the \texttt{Reddit-V2}~\citep{he2023zeroshot, zhu2025rankgrpo} dataset, where we combine recommendations from different Reddit users on the same conversation as multiple ground truths. For validation, we use sampled negatives with $200$ candidates in total for efficient model selection, and the test results are reported against the full catalog. We report NDCG@$K$ for $K \in \{5, 15, 20\}$. Full implementation details are provided in Appendix~\ref{sec:imp}.

\begin{figure}[t]
\centering
\includegraphics[width=\linewidth]{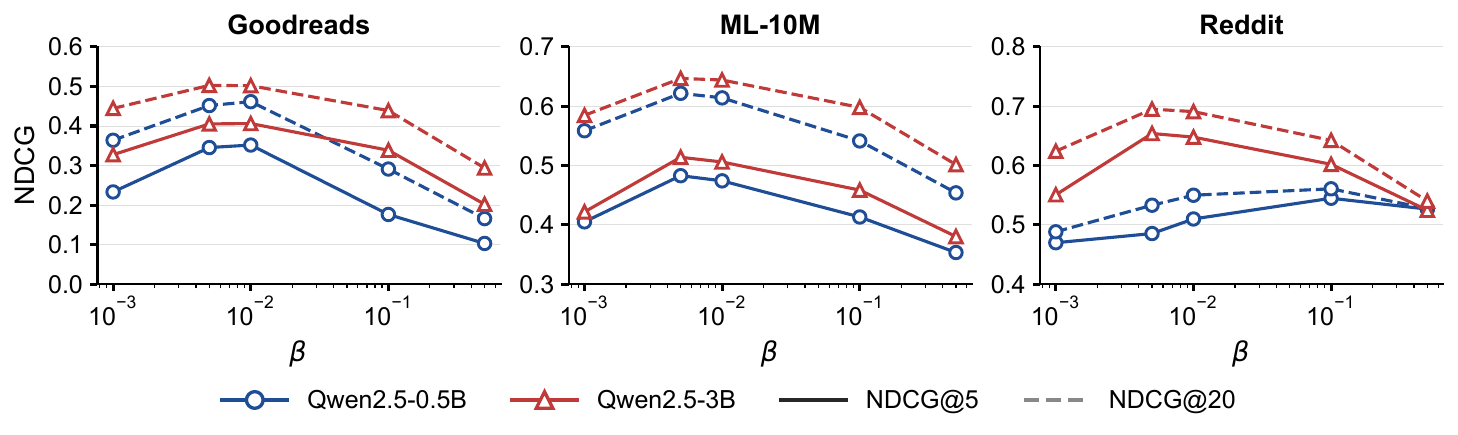}
\vspace{-2mm}
\caption{Validation NDCG@$5$ (solid) and NDCG@$20$ (dashed) versus the regularization strength $\beta$.}
\label{fig:hparam_sensitivity_val}
\vspace{-3mm}
\end{figure}

\subsection{Training Dynamics and Hyperparameter Analysis}
\label{sec:exp_hyper}

We first study the effects of the regularization strength $\beta$ in Eq.~\eqref{eq:multdpo_rlhf} on the validation sets, summarized in Figure~\ref{fig:hparam_sensitivity_val}. Due to space limitations and constraint of computational budgets, we present the results with the \texttt{Qwen-2.5-0.5B-Instruct}~\citep{yang2025qwen} and \texttt{Qwen-2.5-3B-Instruct} backbones, while we use the insight to select $\beta$ to run Mult-DPO on the Goodreads dataset with $\{0.5B, 1.5B, 3B, 7B\}$ backbones in \ref{sec:exp_more_backbones} to show its generalization ability w.r.t. model scale. We observe that too small a $\beta$ lets $\pi_\theta$ drift too far from a strong $\pi_{\mathrm{ref}}$, while too large a $\beta$ over-regularizes and erodes the alignment signal. The optimal $\beta$ is consistently smaller on \texttt{MovieLens-10M} and \texttt{Goodreads} than on \texttt{Reddit-V2}. This may be because the reference backbone provides a stronger initialization for conversational recommendation, where the dialogue context can be exploited via language understanding~\citep{he2023zeroshot}, than for general recommendation, where collaborative-filtering signals that are unavailable to pretrained LLMs play a crucial role~\citep{zhu2025collaborative}.

\subsection{Comparison with Baselines}
\label{sec:exp_main}

We then evaluate Mult-DPO against representative baselines on the test set, where the best $\beta$ for each DPO-style method is selected on the validation set as specified in Section~\ref{sec:exp_hyper}. All DPO-style baselines share the same reference backbone and training protocol as Mult-DPO for fair comparison.

\vspace{1mm}
\noindent \textbf{Baselines.}
We compare Mult-DPO against two categories of LLM-based recommenders: SFT-only baselines (BigRec~\citep{bao2023tallrec} and $\text{D}^{3}$~\citep{bao2024decoding}) that fine-tune the LLM on context--item demonstrations, and DPO-style alignment baselines (vanilla DPO~\citep{rafailov2023direct}, DMPO~\citep{bai2024dmpo}, S-DPO~\citep{chen2024softmax}, and LiPO (BT)~\citep{liu2024lipo}). LiPO is a family of learning-to-rank algorithms; the most relevant variant, LiPO (BT), sums the pairwise Bradley--Terry losses over all positive--negative pairs in a single step. Full details are provided in Appendix~\ref{sec:baselines}. For $\text{D}^{3}$ on conversational recommendation, we follow~\citet{zhu2024collaborative} to extract collaborative-filtering signals and train an EASE~\citep{steck2019embarrassingly} model as the text-free assistant.

\begin{table*}[t]
\caption{Comparison between Mult-DPO and DPO-style baselines on the test sets of \texttt{Goodreads}, \texttt{MovieLens-10M}, and \texttt{Reddit-V2}. Evaluation is conducted with the entire catalog as the candidates. Most results are significant as standard error is between 0.0020-0.0035 for N@5, N@15 and N@20.} 
\vspace{2mm}
\centering
\fontsize{6.8pt}{9.5pt}\selectfont
\setlength{\tabcolsep}{4pt}
\begin{tabular}{l|ccc|ccc|ccc}
\toprule
\multirow{2}{*}{\textbf{Method}}
& \multicolumn{3}{c|}{\texttt{Goodreads}}
& \multicolumn{3}{c|}{\texttt{MovieLens-10M}}
& \multicolumn{3}{c}{\texttt{Reddit-V2}} \\
\cmidrule(lr){2-4} \cmidrule(lr){5-7} \cmidrule(lr){8-10}
& \textbf{N@5} & \textbf{N@15} & \textbf{N@20}
& \textbf{N@5} & \textbf{N@15} & \textbf{N@20}
& \textbf{N@5} & \textbf{N@15} & \textbf{N@20} \\
\midrule
\rowcolor{gray!15} \texttt{Qwen2.5-0.5B-Instruct} (zero-shot) & 0.0085 & 0.0090 & 0.0090 & 0.0136 & 0.0156 & 0.0156 & 0.0175 & 0.0207 & 0.0216 \\
\multicolumn{10}{c}{\textit{SFT-based methods}} \\
\ \; + BigRec~\citep{bao2023tallrec}          & 0.0776 & 0.1069 & 0.1177 & \textbf{0.0657} & 0.0884 & 0.0962 & 0.1043 & 0.1056 & 0.1083  \\
\ \; + $\text{D}^{3}$~\citep{bao2024decoding}             & 0.0818 & 0.1122 & 0.1210 & 0.0612 & 0.0881 & 0.0948 & 0.1015 & 0.1037 & 0.1064\\
\multicolumn{10}{c}{\textit{DPO-based methods}} \\
\ \; + Vanilla DPO~\citep{rafailov2023direct} & 0.0389 & 0.0558 & 0.0622 & 0.0426 & 0.0710 & 0.0834 & 0.0816 & 0.0819 & 0.0838 \\
\ \; + DMPO~\citep{bai2024dmpo}               & 0.0586 & 0.0805 & 0.0845 & 0.0461 & 0.0701 & 0.0758 & 0.0875 & 0.0893 & 0.0912 \\
\ \; + S-DPO~\citep{chen2024softmax}          & 0.0762 & 0.1105 & 0.1192 & 0.0592 & 0.0920 & 0.1049 & 0.0931 & 0.0938 & 0.0985 \\
\ \; + LiPO (BT)~\citep{liu2024lipo}               & 0.0862 & 0.1198 & 0.1294 & 0.0622 & 0.0980 & 0.1100 & 0.0963 & 0.1020 & 0.1060 \\
\rowcolor{green!10} \ \; + Mult-DPO (ours)    & \textbf{0.0947} & \textbf{0.1292} & \textbf{0.1406} & 0.0650 & \textbf{0.1001} & \textbf{0.1103} & \textbf{0.1097} & \textbf{0.1101} & \textbf{0.1154} \\
\midrule
\rowcolor{gray!15} \texttt{Qwen2.5-3B-Instruct} (zero-shot)     & 0.0149 & 0.0184 & 0.0184 & 0.0232 & 0.0340 & 0.0359 & 0.0633 & 0.0617 & 0.0641 \\
\multicolumn{10}{c}{\textit{SFT-based methods}} \\
\ \; + BigRec~\citep{bao2023tallrec}          & 0.1109 & 0.1527 & 0.1682 & 0.0747 & 0.1069 & 0.1180 & 0.1228 & 0.1247 & 0.1364  \\
\ \; + $\text{D}^{3}$~\citep{bao2024decoding} & 0.1254 & 0.1531 & 0.1678 & 0.0710 & 0.1038 & 0.1159 & 0.1195 & 0.1214 & 0.1332 \\
\multicolumn{10}{c}{\textit{DPO-based methods}} \\
\ \; + Vanilla DPO~\citep{rafailov2023direct} & 0.0870 & 0.1120 & 0.1202 & 0.0559 & 0.0806 & 0.0903 & 0.0915 & 0.0960 & 0.0998 \\
\ \; + DMPO~\citep{bai2024dmpo}               & 0.0932 & 0.1200 & 0.1359 & 0.0562 & 0.0875 & 0.0938 & 0.0981 & 0.1009 & 0.1032 \\
\ \; + S-DPO~\citep{chen2024softmax}          & 0.1181 & 0.1586 & 0.1693 & 0.0631 & 0.0989 & 0.1122 & 0.1043 & 0.1132 & 0.1219 \\
\ \; + LiPO (BT)~\citep{liu2024lipo}               & 0.1252 & 0.1611 & 0.1731 & 0.0672 & 0.1046 & 0.1185 & 0.1147 & 0.1234 & 0.1329 \\
\rowcolor{green!10} \ \; + Mult-DPO (ours)    & \textbf{0.1288} & \textbf{0.1678} & \textbf{0.1785} & \textbf{0.0751} & \textbf{0.1155} & \textbf{0.1300} & \textbf{0.1369} & \textbf{0.1431} & \textbf{0.1503} \\
\bottomrule
\end{tabular}
\label{tab:main_results}
\vspace{-3mm}
\end{table*}

\vspace{1mm}
\noindent \textbf{Results.}
From Table~\ref{tab:main_results} we can find that DMPO extends vanilla DPO with multiple negatives for one positive item, but the arithmetic mean it takes for the negatives inside the BT sigmoid is an ad-hoc construction without a coherent ranking-likelihood interpretation. S-DPO replaces this mean with the softmax derived as the closed-form $k\!=\!1$ case of the marginalized PL likelihood (a special case of Theorem~\ref{thm:multdpo_tightness}), providing a principled treatment of multiple negatives that respects the underlying preference model. LiPO (BT) further admits multi-positive supervision by summing up pairwise losses, yet it still factorizes the set-wise event into independent BT pairs and discards the constraint that the positives jointly dominate the negative set. In contrast, Mult-DPO fundamentally removes this spurious independence: we note that under the multinomial surrogate every positive is contrasted against the same cumulative negative weight $B$, which preserves the joint set-wise structure and is provably an upper bound on the marginalized PL-DPO loss (Corollary~\ref{cor:multdpo_dpo_bound}). The margin over LiPO (BT) is most pronounced on \texttt{Reddit-V2}, where multi-positive ground truths are densest, and widens with backbone scale, suggesting that the joint signal pays off most when the policy has the capacity to exploit it. Mult-DPO also surpasses BigRec and $\text{D}^{3}$ from the same SFT initialization, indicating that set-wise alignment supplies supervision that demonstration matching cannot recover. Beyond accuracy, we note that compared with the strongest baseline LiPO (BT), in Mult-DPO, the sharing $B$ across positives reduces the loss-aggregation cost from $\mathcal{O}((K-k)\cdot k)$ to $\mathcal{O}(K)$.

\subsection{On Upper Bound of MN-DPO to Marginalized PL-DPO and Its Tightness} 

\begin{figure}[t]
\centering
\includegraphics[width=0.92\textwidth]{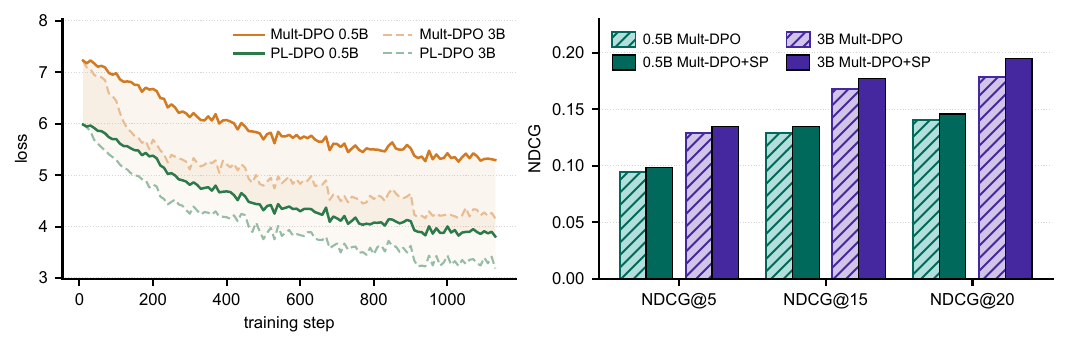}\vspace{-2mm}
\caption{\textbf{Left}: Mult-DPO loss and exact marginalized PL-DPO loss; \textbf{Right}: NDCG on the \texttt{Goodreads} dataset of Mult-DPO with SPRec-style epoch-level dynamic hard negatives.}
\label{fig:tightness}
\vspace{-4mm}
\end{figure}

For \textbf{RQ2}, we first restrict the training set to samples with at most three positives, for which the exact marginalized PL-DPO loss can be computed. As shown in Fig.~\ref{fig:tightness} (left), on this filtered subset, the training dynamics of Mult-DPO and the exact marginalized PL-DPO verify the upper-bound relation in Corollary~\ref{cor:multdpo_dpo_bound}. However, this restriction removes most training instances, making the resulting test-set evaluation uninformative. To further examine whether harder negatives tighten the bound and thereby improve alignment with multi-positive supervision, we consider dynamic negative sampling. Since the hardness of a negative item is policy-dependent and changes throughout training, we adapt  SPRec~\citep{gao2025sprec} and resample negatives at the epoch level. Concretely, at the beginning of each epoch, negatives are drawn from a temperature-scaled categorical distribution induced by the current policy weights with temperature set to 0.1. Fig.~\ref{fig:tightness} (right) reports the  test performance, where we find that introducing harder negatives can indeed improve the alignment ability of Mult-DPO.

\subsection{Extension to Multi-Level Preference on the MovieLens Dataset}
\label{sec:exp_multlevel}

\begin{wrapfigure}{r}{0.48\textwidth}
\vspace{-4mm}
\centering
\includegraphics[width=0.48\textwidth]{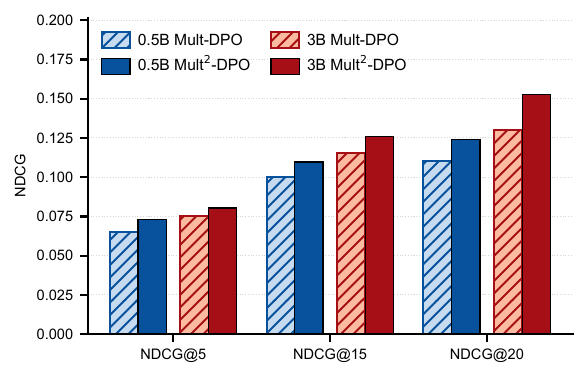}
\vspace{-6mm}
\caption{Comparison of Mult-DPO and \ Mult$^{2}$-DPO NDCG on \texttt{MovieLens-10M} dataset.}
\label{fig:multi2_vs_multi}
\vspace{-3mm}
\end{wrapfigure}
To answer \textbf{RQ3}, on \texttt{MovieLens-10M}, we partition items by their ratings into $G = 4$ ordered preference groups, with rating $=5$ (i.e., the positive set in the two-set case) as the group with the highest preference, and with randomly sampled unrated items appended after the high-rated groups (rating $> 3$) to keep the candidate set comparable to the binary setting of Section \ref{sec:exp_main}. Mult$^{2}$-DPO aggregates the per-boundary multinomial losses across the three group boundaries via Eq.~\eqref{eq:multdpo_group_final}, and is compared against the binary Mult-DPO baseline that merges the four groups into one positive-negative split. Figure~\ref{fig:multi2_vs_multi} shows that Mult$^{2}$-DPO outperforms the binary counterpart at every cutoff, with a $\sim\!12\%$ gain in NDCG at the 0.5B backbone ($0.0732$ vs.\ $0.0650$ at NDCG@5). In addition, the improvement can be generalized to the 3B backbone. This further confirms that Mult$^{2}$-DPO provides a stronger alignment signal by preserving the richer preference structure of explicit ratings.

\section{Conclusion}
\label{sec:conclusion}

In this paper, we presented \textbf{Mult-DPO}, a new DPO-style framework for aligning LLM-based RSs with preferences over multiple candidates. The key contribution is a multinomial surrogate event likelihood that yields a tractable closed-form DPO-style loss that is a provable upper bound on the otherwise intractable marginalized Plackett-Luce loss. We also characterize the tightness of this bound in terms of the number of positives and the cumulative policy-induced weights of the selected negatives. Empirically, Mult-DPO and its multi-level extension, i.e., Mult$^{2}$-DPO, consistently outperform various  DPO-based baselines on both general and conversational recommendation benchmarks, suggesting that explicitly modeling the joint set-wise structure of user feedback is both effective and efficient for preference-based alignment in LLM-based recommendation tasks.

\bibliographystyle{plainnat}
\bibliography{neurips_2026}

\newpage
\appendix

\textbf{\large Appendix}
\vspace{-2mm}

\section{Proofs}
\label{app:proofs}
\label{appendix:mult-dpo}

In this section, we provide the proofs and derivations for Section~\ref{sec:method} in the main paper, including the \textbf{\textit{(i)}} marginalized Plackett-Luce likelihood, \textbf{\textit{(ii)}} its inclusion-exclusion form, \textbf{\textit{(iii)}} the multinomial surrogate likelihood, \textbf{\textit{(iv)}} the lower-bound relation between MN surrogate and PL likelihoods, \textbf{\textit{(v)}} the tightness analysis, and \textbf{\textit{(vi)}} the extension of Mult-DPO to multi-level preference alignment of LLMs.

\subsection{Full-Rank Plackett-Luce Model and Its Recursions.}

Recall that in the main paper, we use $\mathcal{E}$ to denote the candidate set associated with a recommendation-related context $x$. Let $K := |\mathcal{E}|$ and let $\tau$ be a permutation of $\{1, \ldots, K\}$ representing a full ranking of the items in $\mathcal{E}$. We write $e_{\tau(t)}$ for the item placed at rank $t$ under $\tau$. For any $t \in \{1, \ldots, K\}$, define the prefix ranking of length $t$ and the suffix ranking of the remaining items as
\[
\tau_{1:t} := (\tau(1), \tau(2), \ldots, \tau(t)), \ \
\tau_{t+1:K} := (\tau(t+1), \tau(t+2), \ldots, \tau(K)).
\]
The suffix candidate set after selecting the first $t$ ranked items is defined as
\[
\mathcal{E}_{t+1} := \mathcal{E} \setminus \{e_{\tau(1)}, \ldots, e_{\tau(t)}\}.
\]
The Plackett-Luce (PL) \citep{plackett1975analysis,luce1959individual} model assigns the likelihood to a full ranking $\tau$:
\begin{subequations}\label{eq:PL-full}
\begin{align}
p_{\mathrm{PL}}(\tau \mid x, \mathcal{E}; w)
&=
\prod_{j=1}^{K}
\frac{w(e_{\tau(j)})}{\sum_{l=j}^{K} w(e_{\tau(l)})}
\label{eq:PL-full:a}
\\
&=
\frac{w(e_{\tau(1)})}{\sum_{e \in \mathcal{E}} w(e)}
\cdot
\prod_{j=2}^{K}
\frac{w(e_{\tau(j)})}{\sum_{l=j}^{K} w(e_{\tau(l)})}
\nonumber\\
&=
\frac{w(e_{\tau(1)})}{\sum_{e \in \mathcal{E}} w(e)}
\cdot
p_{\mathrm{PL}}\!\left(
\tau_{2:K} \mid x, \mathcal{E}_{2}; w
\right)
\label{eq:PL-full:c}
\\
&=
\left(
\prod_{t=1}^{k}
\frac{w(e_{\tau(t)})}{\sum_{e \in \mathcal{E}_{t}} w(e)}
\right)
\cdot
\prod_{j=k+1}^{K}
\frac{w(e_{\tau(j)})}{\sum_{l=j}^{K} w(e_{\tau(l)})}
\nonumber\\
&=
\left(
\prod_{t=1}^{k}
\frac{w(e_{\tau(t)})}{\sum_{e \in \mathcal{E}_{t}} w(e)}
\right)
\cdot
p_{\mathrm{PL}}\!\left(
\tau_{k+1:K} \mid x, \mathcal{E}_{k+1}; w
\right),
\label{eq:PL-full:e}
\end{align}
\end{subequations}
where Eqs.~\eqref{eq:PL-full:c} and~\eqref{eq:PL-full:e} are the one-step and $k$-step recursions, respectively.

\vspace{1mm}

\begin{lemma}[PL normalization on a finite candidate set]
\label{lem:PL-normalization-app}
Let $\mathcal{E} = \{e_1, \ldots, e_K\}$ be a finite candidate set and let $S_K$ denote the set of all permutations of $\{1, \ldots, K\}$. Then the PL likelihood in Eq.~\eqref{eq:PL-full:a} defines a valid distribution over full rankings:
\begin{equation}
\sum_{\tau \in S_K} p_{\mathrm{PL}}(\tau \mid x, \mathcal{E}; w) = 1.
\label{eq:PL-normalization-app}
\end{equation}
\end{lemma}

\begin{proof}
When $K = 1$, Eq.~\eqref{eq:PL-normalization-app} trivially holds. For arbitrary $K$, we have
\begin{subequations}\label{eq:plnorm-app}
\begin{align}
\sum_{\tau \in S_K} p_{\mathrm{PL}}(\tau \mid x, \mathcal{E}; w)
&=
\sum_{i=1}^{K}\ \sum_{\rho \in \mathrm{Perm}(\mathcal{E} \setminus \{e_i\})}
p_{\mathrm{PL}}\bigl([e_i, \rho] \mid x, \mathcal{E}; w\bigr)
\label{eq:plnorm-app:b}\\
&=
\sum_{i=1}^{K}\ \sum_{\rho \in \mathrm{Perm}(\mathcal{E} \setminus \{e_i\})}
\frac{w(e_i)}{\sum_{i=1}^{K} w(e_i)}
\cdot
p_{\mathrm{PL}}\bigl(\rho \mid x, \mathcal{E} \setminus \{e_i\}; w\bigr)
\label{eq:plnorm-app:c}\\
&=
\sum_{i=1}^{K}
\frac{w(e_i)}{\sum_{i=1}^{K} w(e_i)}
\left(
\sum_{\rho \in \mathrm{Perm}(\mathcal{E} \setminus \{e_i\})}
p_{\mathrm{PL}}\bigl(\rho \mid x, \mathcal{E} \setminus \{e_i\}; w\bigr)
\right)
\nonumber\\
&=
\sum_{i=1}^{K} \frac{w(e_i)}{\sum_{i=1}^{K} w(e_i)}
=
1,
\label{eq:plnorm-app:e}
\end{align}
\end{subequations}
\vspace{0.22mm}
where the last step follows by induction on $K$.
\end{proof}

\vspace{2mm}

\subsection{Proof of the Marginalized PL Likelihood in Eq.~\eqref{eq:multdpo_pl_likelihood}.}
\label{sec:proof_marginalized_pl}

Recall that the positives are $\mathcal{E}^{p} = \{e_1, \ldots, e_k\}$ and the negatives are $\mathcal{E}^{d} = \{e_{k+1}, \ldots, e_K\}$. Let $S_k$ denote the set of permutations of $\{1, \ldots, k\}$ for the positive items. Then
\begin{subequations}\label{eq:kfac-app}
\begin{align}
p_{\mathrm{PL}}(\Omega_x \mid x, \mathcal{E}; w)
&=
\sum_{\rho \in S_k}
\ \sum_{\tau_{k+1:K}}
p_{\mathrm{PL}}\!\left(
[\rho, \tau_{k+1:K}] \mid x, \mathcal{E}; w
\right)
\label{eq:kfac-app:a}\\
&=
\sum_{\rho \in S_k} \sum_{\tau_{k+1:K}}
\left(
\prod_{t=1}^{k}
\frac{w(e_{\rho(t)})}
{\sum_{e \in \mathcal{E} \setminus \{e_{\rho(1)}, \ldots, e_{\rho(t-1)}\}} w(e)}
\right)
p_{\mathrm{PL}}\!\left(
\tau_{k+1:K} \mid x, \mathcal{E}^{d}; w
\right)
\label{eq:kfac-app:b}\\
&=
\sum_{\rho \in S_k}
\left(
\prod_{t=1}^{k}
\frac{w(e_{\rho(t)})}
{\sum_{e \in \mathcal{E} \setminus \{e_{\rho(1)}, \ldots, e_{\rho(t-1)}\}} w(e)}
\right)
\sum_{\tau_{k+1:K}}
p_{\mathrm{PL}}\!\left(
\tau_{k+1:K} \mid x, \mathcal{E}^{d}; w
\right)
\nonumber\\
&=
\sum_{\rho \in S_k}
\prod_{t=1}^{k}
\frac{w(e_{\rho(t)})}
{\sum_{e \in \mathcal{E} \setminus \{e_{\rho(1)}, \ldots, e_{\rho(t-1)}\}} w(e)}
\label{eq:kfac-app:c}\\
&=
\sum_{\rho \in S_k}
\prod_{t=1}^{k}
\frac{w(e_{\rho(t)})}
{B + \sum_{j=t}^{k} w(e_{\rho(j)})}.
\label{eq:kfac-app:d}
\end{align}
\end{subequations}
In the last step, the remaining candidate set at rank $t$ consists of the unselected positives $\{e_{\rho(t)}, \ldots, e_{\rho(k)}\}$ and all negatives $\mathcal{E}^{d}$, so the denominator decomposes into $B + \sum_{j=t}^{k} w(e_{\rho(j)})$. Note that the marginalization over the latent rankings of the negatives vanishes from step~\eqref{eq:kfac-app:b} to step~\eqref{eq:kfac-app:c} because the PL likelihood over the suffix ranking of $\mathcal{E}^{d}$ sums to one by Lemma~\ref{lem:PL-normalization-app}.

\vspace{2mm}

\subsection{Proof of the Inclusion-Exclusion Form for the Marginalized PL Likelihood.}\label{app:pl-ie}

\begin{lemma}[Exponential-race representation]
\label{lem:race-app}
Let $\{T_e\}_{e \in \mathcal{E}}$ be independent random variables with $T_e \sim \Exp(w(e))$. Ranking items by increasing $T_e$ induces the PL distribution in Eq.~\eqref{eq:PL-full}.
\end{lemma}

\begin{proof}
This is a standard characterization of the Plackett-Luce model. The minimum exponential clock is selected first with probability proportional to its rate, and the memoryless property gives the same rule recursively for the remaining items.
\end{proof}

\noindent Using Lemma~\ref{lem:race-app}, we have
\begin{subequations}\label{eq:ie-app}
\begin{align}
p_{\mathrm{PL}}(\Omega_x \mid x, \mathcal{E}; w)
&=
\Pr\!\left(
\max_{e \in \mathcal{E}^{p}} T_e
<
\min_{e \in \mathcal{E}^{d}} T_e
\right)
\label{eq:ie-app:a}\\
&=
\Pr\!\left(
T_e < U \;\; \forall e \in \mathcal{E}^{p}
\right),
\ \
U := \min_{e \in \mathcal{E}^{d}} T_e
\nonumber\\
&=
\int_{0}^{\infty}
\Pr\!\left(
T_e < u \;\; \forall e \in \mathcal{E}^{p}
\right)
\, f_U(u)\,du
\label{eq:ie-app:c}\\
&=
\int_{0}^{\infty}
B e^{-Bu}
\prod_{e \in \mathcal{E}^{p}}
\Pr(T_e < u)\,du
\label{eq:ie-app:d}\\
&=
\int_{0}^{\infty}
B e^{-Bu}
\prod_{e \in \mathcal{E}^{p}}
\bigl(1 - e^{-w(e) u}\bigr)\,du
\label{eq:ie-app:e}\\
&=
\int_{0}^{\infty}
B e^{-Bu}
\sum_{S \subseteq \mathcal{E}^{p}}
(-1)^{|S|}
\prod_{e \in S} e^{-w(e) u}\,du
\label{eq:ie-app:f}\\
&=
\sum_{S \subseteq \mathcal{E}^{p}}
(-1)^{|S|}
\int_{0}^{\infty}
B \exp\!\left(-u\Bigl(B + \sum_{e \in S} w(e)\Bigr)\right)\,du
\nonumber\\
&=
\sum_{S \subseteq \mathcal{E}^{p}}
(-1)^{|S|}
\frac{B}{B + \sum_{e \in S} w(e)}.
\label{eq:multdpo_pl_ie}
\end{align}
\end{subequations}
This proves the inclusion-exclusion form for the marginalized PL likelihood, which reduces the number of terms in naive marginalization (i.e., Eq. (\ref{eq:multdpo_pl_likelihood})) from $k!$ to $2^{k}$. However, it still remains intractable for moderate $k$, motivating the multinomial surrogate developed in Mult-DPO.

\vspace{2mm}

\subsection{Derivation of the Multinomial Surrogate Likelihood in Eq.~\eqref{eq:multdpo_mn_likelihood}.}

We formally derive the multinomial surrogate likelihood. Recall from Section~\ref{sec:mult} that the categorical distribution over $\mathcal{E}$ is defined as $p(e \mid x) = w(e \mid x) / W$. The MN surrogate assigns to $\Omega_x$ the probability of the count event induced by $k$ i.i.d.\ draws from $p(\cdot \mid x)$ in which each positive item in $\mathcal{E}^{p}$ is sampled exactly once and no negative item in $\mathcal{E}^{d}$ is sampled. Letting $c(e)$ denote the count of item $e$ across the $k$ draws, this event can be formulated as
\[
c(e) = 1 \ \forall e \in \mathcal{E}^{p}, \ \ c(e) = 0 \ \forall e \in \mathcal{E}^{d}.
\]
Substituting these counts into the multinomial probability mass function then yields
\begin{subequations}\label{eq:MN-app-deriv}
\begin{align}
p_{\mathrm{MN}}(\Omega_x \mid x, \mathcal{E}; w)
&=
\frac{k!}{\prod_{e \in \mathcal{E}} c(e)!}
\prod_{e \in \mathcal{E}} p(e \mid x)^{c(e)}
\label{eq:MN-app:a}\\
&=
k!
\prod_{e \in \mathcal{E}^{p}} p(e \mid x)
\label{eq:MN-app:b}\\
&=
k!
\prod_{e \in \mathcal{E}^{p}}
\frac{w(e \mid x)}{W}.
\label{eq:MN-app:c}
\end{align}
\end{subequations}

\vspace{2mm}

\subsection{Proof of Theorem~\ref{thm:multdpo_lb}.}
\label{sec:proof1}

We now prove the main theorem of the paper, i.e., the multinomial surrogate likelihood is a pointwise lower bound on the intractable marginalized PL likelihood. Recall from Eq.~\eqref{eq:multdpo_pl_likelihood} that
\[
p_{\mathrm{PL}}(\Omega_x \mid x, \mathcal{E}; w)
=
\sum_{\rho \in S_k}
\prod_{t=1}^{k}
\frac{w(e_{\rho(t)} \mid x)}{B + \sum_{j=t}^{k} w(e_{\rho(j)} \mid x)}.
\]
Since the per-rank denominator satisfies $B + \sum_{j=t}^{k} w(e_{\rho(j)} \mid x) \le B + \sum_{e \in \mathcal{E}^{p}} w(e \mid x) = W$, we obtain the following chain of pointwise inequalities for any $\rho \in S_k$ and any $t \in \{1, \ldots, k\}$:
\begin{align*}
\frac{1}{B + \sum_{j=t}^{k} w(e_{\rho(j)} \mid x)}
&\ge
\frac{1}{W} \\
\Longrightarrow \quad
\frac{w(e_{\rho(t)} \mid x)}{B + \sum_{j=t}^{k} w(e_{\rho(j)} \mid x)}
&\ge
\frac{w(e_{\rho(t)} \mid x)}{W} \\
\Longrightarrow \quad
\prod_{t=1}^{k}
\frac{w(e_{\rho(t)} \mid x)}{B + \sum_{j=t}^{k} w(e_{\rho(j)} \mid x)}
&\ge
\prod_{t=1}^{k}
\frac{w(e_{\rho(t)} \mid x)}{W}.
\end{align*}
Summing over $\rho \in S_k$ gives
\[
p_{\mathrm{PL}}(\Omega_x \mid x, \mathcal{E}; w)
\ge
\sum_{\rho \in S_k}
\prod_{t=1}^{k}
\frac{w(e_{\rho(t)} \mid x)}{W}
=
k! \prod_{e \in \mathcal{E}^{p}} \frac{w(e \mid x)}{W}
=
p_{\mathrm{MN}}(\Omega_x \mid x, \mathcal{E}; w).
\]

\vspace{2mm}

\subsection{Proof of Theorem~\ref{thm:multdpo_tightness}.}
\label{sec:proof2}

We then characterize the tightness of the lower bound, i.e., the ratio between the marginalized PL likelihood and the multinomial surrogate likelihood is sandwiched between $1$ and $(1 + A/B)^{k-1}$. Specifically, we first fix an arbitrary $\rho \in S_k$ and define the prefix sums of positive weights
\[
H_{t-1}(\rho)
:=
\sum_{j=1}^{t-1} w(e_{\rho(j)} \mid x),
\ \ t = 2, \ldots, k.
\]
For the PL model, the denominator at rank $t$ is
\[
D_t(\rho)
=
B + \sum_{j=t}^{k} w(e_{\rho(j)} \mid x)
=
W - H_{t-1}(\rho).
\]
Therefore, the ratio between the per-$\rho$ PL term and the corresponding MN term satisfies
\begin{subequations}\label{eq:tight-app}
\begin{align}
\frac{
\prod_{t=1}^{k}
\frac{w(e_{\rho(t)} \mid x)}{D_t(\rho)}
}{
\prod_{t=1}^{k}
\frac{w(e_{\rho(t)} \mid x)}{W}
}
&=
\prod_{t=2}^{k}
\frac{W}{W - H_{t-1}(\rho)}
\label{eq:tight-app:a}\\
&=
\prod_{t=2}^{k}
\frac{1}{1 - \frac{H_{t-1}(\rho)}{W}}
\label{eq:tight-app:b}\\
&\le
\prod_{t=2}^{k}
\frac{1}{1 - \frac{A}{W}}
=
\left(\frac{W}{W - A}\right)^{k-1}
=
\left(1 + \frac{A}{B}\right)^{k-1}.
\label{eq:tight-app:c}
\end{align}
\end{subequations}
Since this bound holds uniformly over all $\rho \in S_k$, summing over $\rho$ yields
\[
1
\le
\frac{p_{\mathrm{PL}}(\Omega_x \mid x, \mathcal{E}; w)}
     {p_{\mathrm{MN}}(\Omega_x \mid x, \mathcal{E}; w)}
\le
\left(1 + \frac{A}{B}\right)^{k-1}.
\]

\vspace{2mm}

\subsection{Derivation for the DPO-style Extension.}
In this part, we connect weights in both the MN surrogate and the marginalized PL event model with the LLM-based RS policy $\pi_{\theta}$ through the analytical solution of the RLHF objective, i.e., $w_{\pi_\theta}(e \mid x)$, so that both bounds hold for the corresponding DPO-style losses. Let $\pi_{\mathrm{ref}}(e \mid x)$ be a reference policy and $\pi_{\theta}(e \mid x)$ a learned policy. The RLHF solution gives weights proportional to $\left(\pi_\theta(e \mid x)/\pi_{\mathrm{ref}}(e \mid x)\right)^{\beta}$, with a context-dependent factor $Z(x)^\beta$ shared by all candidates. Since both PL and MN likelihoods are invariant to multiplying all weights in the same candidate set by a common positive constant, we fix the representative
\[
r_{\pi_\theta}(e \mid x) := \beta \log \frac{\pi_\theta(e \mid x)}{\pi_{\mathrm{ref}}(e \mid x)}, \ \
w_{\pi_\theta}(e \mid x) := \exp(r_{\pi_\theta}(e \mid x))
=
\left(\frac{\pi_\theta(e \mid x)}{\pi_{\mathrm{ref}}(e \mid x)}\right)^\beta.
\]
If we define $p_{\mathrm{PL}}(\Omega_x \mid x, \mathcal{E}; w_{\pi_\theta})$ and $p_{\mathrm{MN}}(\Omega_x \mid x, \mathcal{E}; w_{\pi_\theta})$ by replacing $w(e \mid x)$ with $w_{\pi_\theta}(e \mid x)$ in Eqs.~\eqref{eq:multdpo_pl_likelihood} and~\eqref{eq:multdpo_mn_likelihood}, the lower bound and the tightness bound hold pointwise with $A, B$ replaced by
\[
A_{\pi_\theta} := \sum_{e \in \mathcal{E}^{p}} w_{\pi_\theta}(e \mid x), \ \
B_{\pi_\theta} := \sum_{e \in \mathcal{E}^{d}} w_{\pi_\theta}(e \mid x).
\]
Equivalently, taking negative logarithms reverses the likelihood inequality: the Mult-DPO loss $\mathcal{L}^{\mathrm{Mult\text{-}DPO}}(\pi_\theta) := -\log p_{\mathrm{MN}}(\Omega_x \mid x, \mathcal{E}; w_{\pi_\theta})$ upper bounds the marginalized PL-DPO loss $\mathcal{L}^{\mathrm{PL\text{-}DPO}}(\pi_\theta) := -\log p_{\mathrm{PL}}(\Omega_x \mid x, \mathcal{E}; w_{\pi_\theta})$.

\vspace{2mm}

\subsection{Proof of the Multi-Level PL Factorization in Eq.~\eqref{eq:multdpo_group_pl_factorization}.}
We further extend the theoretical analysis to the multi-level preference alignment setting discussed in Section~\ref{sec:smn}, where user preference takes the form of the multi-level preference event $\Omega^{\mathrm{grp}}_x$ defined in Eq.~\eqref{eq:multdpo_group_event}. We show that $\Omega^{\mathrm{grp}}_x$ factorizes recursively across the boundaries between consecutive groups, reducing the multi-level case to a product of set-wise preference instances each handled by Theorem~\ref{thm:multdpo_lb} and Theorem~\ref{thm:multdpo_tightness}. By definition, the full multi-level preference event satisfies
\[
\Omega^{\mathrm{grp}}_x
=
\bigcap_{g=1}^{G-1} \Omega^{(g)}_x.
\]
By the sequential selection property of PL, once items from the higher-preference groups $\mathcal{E}^{(1)}, \ldots, \mathcal{E}^{(g-1)}$ have been placed, the relative ordering among the remaining groups depends only on the remaining candidate set $\bigcup_{h=g}^{G} \mathcal{E}^{(h)}$, so the events $\Omega^{(g)}_x$ become conditionally independent across $g$. Therefore, the likelihood factorizes recursively as
\begin{align}
p_{\mathrm{PL}}(\Omega^{\mathrm{grp}}_x \mid x, \mathcal{E}; w)
&=
p_{\mathrm{PL}}\!\left(\Omega^{(1)}_x \mid x, \bigcup_{h=1}^{G} \mathcal{E}^{(h)}; w\right)
\cdot
p_{\mathrm{PL}}\!\left(\Omega^{\mathrm{grp}}_x\ \text{restricted to}\ \bigcup_{h=2}^{G} \mathcal{E}^{(h)}\right)
\nonumber\\
&=
p_{\mathrm{PL}}\!\left(\Omega^{(1)}_x \mid x, \bigcup_{h=1}^{G} \mathcal{E}^{(h)}; w\right)
\cdot
p_{\mathrm{PL}}\!\left(\Omega^{(2)}_x \mid x, \bigcup_{h=2}^{G} \mathcal{E}^{(h)}; w\right)
\cdot \ldots
\nonumber\\
&=
\prod_{g=1}^{G-1}
p_{\mathrm{PL}}\!\left(\Omega^{(g)}_x \mid x, \bigcup_{h=g}^{G} \mathcal{E}^{(h)}; w\right),
\label{eq:mpl_seq}
\end{align}
which proves Eq.~\eqref{eq:multdpo_group_pl_factorization}. For each fixed group $g$, the group-$g$-specific preference event $\Omega^{(g)}_x$ is exactly a set-wise preference event of the same form as Eq.~\eqref{eq:set_rank}, with $\mathcal{E}^{(g)}$ playing the role of $\mathcal{E}^{p}$ and $\bigcup_{h=g+1}^{G} \mathcal{E}^{(h)}$ playing the role of $\mathcal{E}^{d}$. Let
\[
A_g := \sum_{e \in \mathcal{E}^{(g)}} w(e \mid x), \quad
B_g := \sum_{h=g+1}^{G}\sum_{e \in \mathcal{E}^{(h)}} w(e \mid x), \quad
W_g := A_g+B_g
\]
denote the current-group, lower-group, and total remaining weights, respectively. Therefore, reusing the binary results yields
\begin{equation}
p_{\mathrm{PL}}\!\left(\Omega^{(g)}_x \mid x, \bigcup_{h=g}^{G} \mathcal{E}^{(h)}; w\right)
=
\sum_{\rho \in S_{k_g}}
\prod_{t=1}^{k_g}
\frac{w(e^{(g)}_{\rho(t)} \mid x)}{B_g + \sum_{j=t}^{k_g} w(e^{(g)}_{\rho(j)} \mid x)},
\label{eq:grp-tier-kfac-app}
\end{equation}
and equivalently
\begin{equation}
p_{\mathrm{PL}}\!\left(\Omega^{(g)}_x \mid x, \bigcup_{h=g}^{G} \mathcal{E}^{(h)}; w\right)
=
\sum_{S \subseteq \mathcal{E}^{(g)}}
(-1)^{|S|}
\frac{B_g}{B_g + \sum_{e \in S} w(e \mid x)}.
\label{eq:grp-tier-ie-app}
\end{equation}

\vspace{2mm}

\subsection{Proof of the Sequential Multinomial Lower Bound and Tightness.}

We finally prove that the DPO loss derived from the sequential multinomial (SMN) surrogate, i.e., the Mult$^{2}$-DPO loss in Eq.~\eqref{eq:multdpo_group_final}, is an upper bound on the DPO loss derived from the marginalized PL reward model in Eq.~\eqref{eq:multdpo_group_pl_factorization} in the multi-level setting. From Eq.~\eqref{eq:multdpo_group_pl_factorization} we have
\begin{equation}
p_{\mathrm{PL}}(\Omega^{\mathrm{grp}}_x \mid x, \mathcal{E}; w)
=
\prod_{g=1}^{G-1}
p_{\mathrm{PL}}\!\left(\Omega^{(g)}_x \mid x, \bigcup_{h=g}^{G} \mathcal{E}^{(h)}; w\right).
\label{eq:smn-proof-app:a}
\end{equation}
For each group $g$, applying Theorem~\ref{thm:multdpo_lb} to the group-$g$-specific preference event $\Omega^{(g)}_x$ on the remaining candidate set $\bigcup_{h=g}^{G} \mathcal{E}^{(h)}$ yields
\begin{equation}
p_{\mathrm{PL}}\!\left(\Omega^{(g)}_x \mid x, \bigcup_{h=g}^{G} \mathcal{E}^{(h)}; w\right)
\ge
k_g! \prod_{e \in \mathcal{E}^{(g)}} \frac{w(e \mid x)}{W_g}.
\label{eq:smn-proof-app:b}
\end{equation}
Here $W_g$ is the total remaining weight from group $g$ onward.
Multiplying Eq.~\eqref{eq:smn-proof-app:b} over $g = 1, \ldots, G-1$ and using Eq.~\eqref{eq:smn-proof-app:a} gives
\[
p_{\mathrm{PL}}(\Omega^{\mathrm{grp}}_x \mid x, \mathcal{E}; w)
\ge
\prod_{g=1}^{G-1}
\left(
k_g! \prod_{e \in \mathcal{E}^{(g)}} \frac{w(e \mid x)}{W_g}
\right)
=
p_{\mathrm{SMN}}(\Omega^{\mathrm{grp}}_x \mid x, \mathcal{E}; w),
\]
which establishes the SMN lower bound on the multi-level PL likelihood.

The tightness bound follows by applying Theorem~\ref{thm:multdpo_tightness} to each factor in Eq.~\eqref{eq:smn-proof-app:a}. For every $g=1,\ldots,G-1$,
\[
1
\le
\frac{
p_{\mathrm{PL}}\!\left(\Omega^{(g)}_x \mid x, \bigcup_{h=g}^{G} \mathcal{E}^{(h)}; w\right)
}{
k_g! \prod_{e \in \mathcal{E}^{(g)}} \frac{w(e \mid x)}{W_g}
}
\le
\left(1+\frac{A_g}{B_g}\right)^{k_g-1}.
\]
Multiplying these inequalities across groups gives
\[
1
\le
\frac{
p_{\mathrm{PL}}(\Omega^{\mathrm{grp}}_x \mid x, \mathcal{E}; w)
}{
p_{\mathrm{SMN}}(\Omega^{\mathrm{grp}}_x \mid x, \mathcal{E}; w)
}
\le
\prod_{g=1}^{G-1}
\left(1+\frac{A_g}{B_g}\right)^{k_g-1}.
\]
Substituting the policy-induced weight $w_{\pi_\theta}(e \mid x)$ and taking negative logarithms flips the likelihood inequality, yielding the Mult$^{2}$-DPO loss in Eq.~\eqref{eq:multdpo_group_final} as a tractable upper bound on the otherwise intractable multi-level marginalized PL-DPO loss and proving Corollary~\ref{cor:multdpo_group_dpo_bound}. Here, we note that when $G = 2$, the bound reduces exactly to the binary Mult-DPO upper bound established in Section~\ref{sec:mult-objective}.

\section{Related Work}
\label{sec:related}

\subsection{LLM-based Recommender Systems}
Recently, LLM-based recommender systems (RSs) have attracted growing attention, where LLMs' knowledge and reasoning ability can be leveraged to generate high-quality recommendations \citep{lin2023howcan,wu2024llmrec}. Early works show that zero-shot LLMs can already serve as effective RSs by ranking candidate items provided in the prompt \citep{hou2024zeroshot, he2023zeroshot}. Post-training approaches further adapt pretrained LLMs to recommendation tasks, encompassing both supervised fine-tuning (SFT) and preference-based alignment. SFT-based methods fine-tune the LLM on prompt-completion pairs of context and item-recommendation demonstrations, with exemplar methods such as TallRec \citep{bao2023tallrec}, BIGRec \citep{bao2025bigrec}, P5 \citep{geng2022p5}, InstructRec \citep{zhang2025instructrec}, and LLaRA \citep{liao2024llara}. To further align LLM-based recommenders with user preferences beyond SFT, more recent works employ preference-based alignment. RL-free DPO-based methods include S-DPO \citep{chen2024softmax}, DMPO \citep{bai2024dmpo}, and SPRec \citep{gao2025sprec} (see Section~\ref{sec:related_work_dpo} for details). A parallel line built on reinforcement learning (RL), particularly group relative policy optimization (GRPO) \citep{shao2024deepseekmath}, generates a list of recommendations via RL: Rec-R1 \citep{lin2025recr1} and Rank-GRPO \citep{zhu2025rankgrpo} sample diverse recommendation trajectories from the LLM and update the policy with corresponding reward signals. However, the quadratic cost of autoregressive item generation limits these methods to generating only a small number of items per query. In contrast, our Mult-DPO natively supports multi-positive supervision and scales to full ranking over the entire catalog.

\subsection{Direct Preference Optimization and Its Role in RS}
\label{sec:related_work_dpo}

Aligning LLMs with user preferences is a critical challenge. Reinforcement learning from human feedback (RLHF) \citep{christiano2017deep, ouyang2022training} first trains a reward model from pairwise preference data and then optimizes the LLM policy against this reward via RL. Direct preference optimization (DPO) \citep{rafailov2023direct} and its variants (including SimPO \citep{meng2024simpo}, IPO \citep{azar2024general}, and KTO \citep{ethayarajh2024kto}) exploit a closed-form connection between the optimal policy and the reward function in RLHF and reduce alignment to maximum likelihood estimation. Adapting DPO to RSs, however, introduces additional challenges due to the multi-item nature of user feedback. Prior work has explored two main strategies. The first generalizes the Bradley-Terry (BT) model underlying vanilla DPO to its listwise counterpart, the Plackett-Luce (PL) model, with methods such as PRO \citep{song2024preference} and KPO \citep{zhang2025kpo} fitting PL likelihoods to fully ordered candidate lists. The second restricts supervision to a single positive paired with multiple negatives: S-DPO \citep{chen2024softmax} adopts a softmax-style loss corresponding to the single-positive special case of the marginalized PL likelihood, DMPO \citep{bai2024dmpo} contrasts the positive against the arithmetic mean of negative log-ratios via an ad-hoc BT sigmoid, and SPRec \citep{gao2025sprec} introduces self-play with model-generated negatives. Nevertheless, neither strategy faithfully captures the joint multi-positive/multi-negative structure of user preference in RSs.

\section{Additional Experiments and Details}
\label{sec:appendix_exp}

\subsection{Implementation Details}
\label{sec:imp}

We run all experiments on the $0.5\text{B}$-$1.5\text{B}$ backbones with $2$ NVIDIA H100 GPUs, and on the $3\text{B}$-$7\text{B}$ backbones with $2$ NVIDIA B200 GPUs. Since the maximum sequence length and per-instance candidate count vary across the three datasets, we adopt different parallelism strategies, ranging from pure data parallelism to ZeRO-2~\citep{rajbhandari2020zero}, while keeping the parallelism strategy identical across baselines for fair comparison. All methods are optimized with AdamW~\citep{loshchilov2019adamw} at a learning rate of $1\mathrm{e}{-6}$. Across both training and inference, we always reuse the KV cache of the prompt across the candidate items associated with the same context, which we find substantially reduces the total compute as detailed in the complexity analysis in Section~\ref{sec:complexity}. 
\subsection{Baselines}
\label{sec:baselines}

The baselines that we include in the main paper are summarized as follows:

\begin{itemize}[leftmargin=*]

\item \textbf{BIGRec}~\citep{bao2025bigrec} is an instruction-tuning framework for LLM-based recommendation that prompts the LLM to generate candidate items from a user's interaction history and then grounds the generated outputs to exact items in the catalog, providing a strong SFT-only baseline.

\item \textbf{$\text{D}^{3}$}~\citep{bao2024decoding} improves alignment at decoding time through a debiasing-diversifying decoding strategy that mitigates amplification bias and homogeneity in LLM-based recommendation outputs, without modifying the underlying training procedure.

\item \textbf{Vanilla DPO}~\citep{rafailov2023direct} is the original DPO objective formulated under the Bradley-Terry (BT) model, where each training instance contains one positive and one negative item.

\item \textbf{DMPO}~\citep{bai2024dmpo} adapts DPO to recommendation by contrasting one positive item against the arithmetic mean of the negative log-ratios within the Bradley-Terry preference model, enabling preference-based alignment with multiple sampled negatives per context.

\item \textbf{S-DPO}~\citep{chen2024softmax} extends vanilla DPO from a single negative to multiple negatives by optimizing a softmax-based preference objective over one preferred item and several dispreferred items, which corresponds to the single-positive special case of Mult-DPO under Theorem~\ref{thm:multdpo_tightness}.

\item \textbf{LiPO (BT)}~\citep{liu2024lipo} casts listwise preference optimization as learning-to-rank over a candidate list. LiPO is a family of such objectives; we adopt its Bradley-Terry variant, LiPO (BT), which sums the pairwise BT losses over all positive-negative pairs at each step.

\end{itemize}

\subsection{Evaluation of Mult-DPO with More Backbone Models}
\label{sec:exp_more_backbones}

\begin{wrapfigure}{r}{0.48\textwidth}
\vspace{-5mm}
\centering
\includegraphics[width=0.48\textwidth]{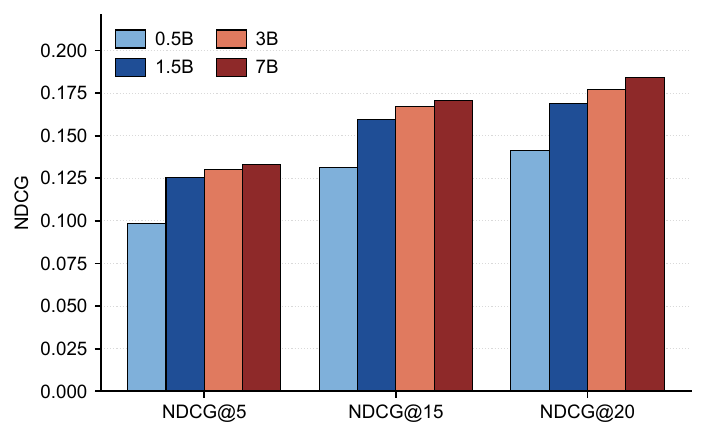}
\caption{Mult-DPO test NDCG on \texttt{Goodreads} across four Qwen2.5 backbones at $\beta{=}0.005$ evaluated on the full item catalog.}
\label{fig:scaling_goodreads}
\vspace{-4mm}
\end{wrapfigure}

To complement the experimental results with 0.5B and 3B backbone models in the main text, we additionally train Mult-DPO on \texttt{Goodreads} with the Qwen2.5-1.5B and Qwen2.5-7B backbones, and report test NDCG@$\{5,15,20\}$ in Fig.~\ref{fig:scaling_goodreads}. Rather than re-running the expensive $\beta$ sweep at every backbone scale, we reuse the optimal $\beta{=}0.005$ identified from the 0.5B and 3B validation sweeps in Section~\ref{sec:exp_hyper} and apply it to the 1.5B and 7B backbone models. Empirically, we find that NDCG climbs steeply from 0.5B to 1.5B backbone model and gradually flattens between 1.5B, 3B, and 7B, with each successive doubling in scale yielding a smaller absolute gain. This pattern indicates that Mult-DPO already extracts most of the available set-wise signal at moderate scale, while the continued improvement at 7B confirms that the benefit of set-wise alignment can be potentially generalized to large backbone regimes.

\section{Limitations}
\label{sec:lim}

We note that the multinomial (MN) construction developed in Mult-DPO is a surrogate event likelihood rather than a normalized ranking distribution. It provides one tractable surrogate for the intractable marginalized Plackett-Luce (PL) likelihood that simultaneously admits a closed-form DPO-style optimization, but it may not be the only or the tightest such surrogate. Because the resulting objective uses a uniform target over positives at the same preference level, it is best matched to settings where such positives can be treated as exchangeable; when positives have substantially different relevance, rating-aware groups or weighted variants may be preferable. Alternative methods could be derived from different angles, e.g., expectation-maximization that introduces latent assignment variables over the unobserved candidate orderings, or a variational lower bound with a flexible variational distribution over the positive permutations. Each alternative offers a different trade-off between tractability and tightness, and a systematic exploration of these directions for new DPO-style alignment with multi-candidate preferences is an interesting direction for future study.

\section{Broad Impact}

Although this paper focuses on recommendation, Mult-DPO is applicable to broader tasks in which set-wise preference data arise, i.e., a set of preferred and dispreferred responses with no reliable internal order, including information retrieval and search ranking with multiple relevant documents, open-domain question answering with multiple acceptable answers, and code generation with several correct programs. As with any preference-based alignment method, the trained policy inherits the preference distribution of its training data, so we recommend pairing Mult-DPO with responsible-deployment practices, such as bias audits and human review, before consumer-facing applications.

\end{document}